\documentclass[twocolumn,english,aps,prd,longbibliography,groupedaddress]{revtex4-1}
\usepackage[T1]{fontenc}
\usepackage[utf8]{inputenc} 
\usepackage{color}
\usepackage{bm}
\usepackage[linesnumbered,ruled,vlined]{algorithm2e}
\usepackage{xcolor}
\usepackage{amsmath}
\usepackage{amsthm}
\usepackage{mathtools}
\usepackage{float}
\usepackage{amssymb}
\usepackage{graphicx}
\usepackage{natbib}
\usepackage{braket}
\usepackage{psfrag}
\usepackage{tikz}
\usepackage[normalem]{ulem}
\usepackage{qcircuit}
\usepackage{siunitx}
\usepackage[colorlinks=true,urlcolor=blue,citecolor=blue,linkcolor=blue]{hyperref}
\pdfmapfile{+sansmathaccent.map}

\usepackage{url}

\theoremstyle{plain}
\newtheorem{theorem}{Theorem}
 
\begin{document}
 

\title{Optimizing Quantum Search with a Binomial Version of Grover's Algorithm}

 
\author{Austin Gilliam}
\affiliation{
JPMorgan Chase}

\author{Marco Pistoia}
\affiliation{
JPMorgan Chase}

\author{Constantin Gonciulea}
\affiliation{
JPMorgan Chase}
 
 
\date{\today}
 
\begin{abstract}
Amplitude Amplification---a key component of Grover's Search algorithm---uses an iterative approach to systematically increase the probability of one or multiple target states. We present novel strategies to enhance the amplification procedure by partitioning the states into classes, whose probabilities are increased at different levels before or during amplification. The partitioning process is based on the binomial distribution.  If the classes to which the search  target states belong are known in advance, the number of iterations in the Amplitude Amplification algorithm can be drastically reduced compared to the standard version. In the more likely case in which the relevant classes are not known in advance, their selection can be configured at run time, or a random approach can be employed, similar to classical algorithms such as binary search. In particular, we apply this method in the context of our previously introduced Quantum Dictionary pattern, where keys and values are encoded in two separate registers, and the value-encoding method is independent of the type of superposition used in the key register. We consider this type of structure to be the natural setup for search.  We confirm the validity of our new approach through experimental results obtained on real quantum hardware, the Honeywell System Model HØ trapped-ion quantum computer.
\end{abstract}
 
\maketitle
 

\section{\label{sec:intro}Introduction}

Quantum Search is a fundamental building block of Quantum Computing, providing a quadratic advantage over its classical counterpart.

 In any context, searching requires a search space and a search condition. In Quantum Computing, a \emph{search space} is defined by the state of a quantum system, and a \emph{search condition} is specified by a \emph{quantum oracle}, which is required to flip the phases of the \emph{target states}, which are those satisfying the underlying condition. Alternatively, the marking can consist of setting certain qubits to fixed values.
 An oracle specific to the target states (search condition) is applied to a superposition state (search space) in order to mark the target states. The Grover Iterate will amplify the probability of the marked state(s).
 
In general, the oracle has to make some assumptions about how the search space is encoded, based on the problem that needs to be solved. In this paper, we are going to look into ways the oracle can adjust the superposition state to amplify the amplitude of the target states in an efficient way.

One can randomly try different types of superposition that favor a set of states, for example those with a given number of $1$s in their binary representations, essentially using a binomial superposition instead of a uniform one. The superposition type can be passed as a parameter to the algorithm.

In this paper, we make the following novel contributions:
\begin{enumerate}
    \item We outline a method for partitioning the basis states into classes, and amplifying the amplitudes of the states in each partition class in a non-uniform way, before repeatedly applying a Grover Iterate specific to a given search target state. When the partition class to which the target state belongs is sufficiently amplified, the number of times the Grover Iterate needs to be applied can be drastically reduced.
    \item This method can provide an additional parameter that can be used to optimize Quantum Search. In particular, in the adaptive version of Grover's Search algorithm, in addition to varying the number of times the Grover Iterate is applied, one can also vary a level of non-uniform amplitude pre-amplification.
    \item We provide the numerical analysis for the partitioning of basis states into classes based on the number of $1$s in their binary representation, which essentially leads to a binomial distribution for the partition amplitudes.
    \item We show how the method can be applied to set search, array element retrieval, and array element search.
    \item We also successfully validate the results on a real quantum computer, and we believe these to be one of the earliest successful quantum search experiments on 4 and 5 qubits using a quantum computer.
\end{enumerate}

The remainder of this paper is structured as follows: Section \ref{sec:grover} goes over the standard, uniform version of Grover's algorithm, which is contrasted with our contribution---the binomial version---in Section \ref{sec:generalized-grover}.  Additional analysis and examples are provided in Section \ref{sec:analysis}.  Section \ref{sec:experimental-results} presents an empirical evaluation of the novel, binomial version of Grover's Search, side by side with the uniform version.  The experimental results are conducted on real hardware.  Finally, Section \ref{sec:conclusion} concludes the paper.

\section{\label{sec:grover}Standard (Uniform) Version of Grover's Algorithm}

\subsection{The Grover Iterate\label{subsec:grover-iterate}}
Given an $n$-qubit quantum system with basis states that correspond to the set $S = \{0, 1,\ldots, 2^n-1\}$, a \emph{predicate function} $f \colon S \to \{0,1\}$, and a set $E \subseteq \{\ket{0},\ket{1},\ldots,\ket{2^n-1}\}$ of \emph{good} basis states that $f$ maps to $1$, the \emph{Grover Iterate} $G$ is comprised of three components~\cite{Brassard2000}:
\begin{enumerate}  
\item A unitary \emph{state-preparation operator} $A$
\item An \emph{oracle} $O$ corresponding to $f$ that recognizes all the elements in $E$ and multiplies their amplitudes by $-1$. For example:
\begin{eqnarray}
OA\ket{0}_n = \sum_{s \notin E} a(s)\ket{s}_n - \sum_{s \in E} a(s)\ket{s}_n
\end{eqnarray}
\item The \emph{diffusion operator} $D$, which multiplies the amplitude of the $\ket{0}_n$ state by $-1$
\end{enumerate}
Given the above, the Grover Iterate is defined as  $G = -ADA^{\dagger}O$. For the \emph{normalized states}, defined as follows:
\begin{eqnarray}
\ket{X_0} = \frac{1}{\sqrt{p(E{'})}}\sum_{s \in E{'}} a(s)\ket{s}
\label{eq:x0}
\end{eqnarray}
\begin{eqnarray}
\ket{X_1} = \frac{1}{\sqrt{p(E)}}\sum_{s \in E} a(s)\ket{s}
\label{eq:x1}
\end{eqnarray}
where $E{'}=S \setminus E$ is the complement of $E$, after applying $j$ iterations of $G$ to $A\ket{0}$, we obtain:
\begin{equation}
\begin{aligned}[b]
G^{j}A\ket{0} &= \cos{((2j+1)\theta)}\ket{X_0} + \\
&\quad\sin{((2j+1)\theta)}\ket{X_1}
\label{eq:k-applications-grover}
\end{aligned}
\end{equation}
where $\theta$ is the angle in $[0, \pi/2]$ that satisfies $\sin^2{\theta=p(E)}$.

\subsection{Grover's Search Algorithm \label{subsec:grover-search-algorithm}}

The original Grover's Search algorithm was formulated as a particular case of the formalism in Section~\ref{subsec:grover-iterate}, where $A=H^{\otimes n}$ creates an equal superposition, and there is only one search target, i.e. $E$ consists of a single element, and $p(E)=\frac{1}{2^n}$.

\begin{theorem}[Uniform Amplification]
The amplitude of a target state marked by an oracle $O$ after applying $j \geq 0$ Grover iterations $G = -H^{\otimes n}DH^{\otimes n}O$ is
\begin{equation}
\sin{((2j+1)\theta)}
\label{eq:amplitude-target-state}
\end{equation}
where
\begin{equation}
\theta = \arcsin\left(\sqrt{\frac{1}{2^n}}\right) .
\label{eq:amplitude-target-state-arcsin}
\end{equation}

\end{theorem}

On the unit circle, the amplitude of the target state starts at $\sin\theta$ and becomes $\sin{((2j+1)\theta)}$ after $j$ iterations, thus approaching $1$ as $(2j+1)\theta$ approaches $\pi/2$.

Note that we can remove the negative sign in the definition of the Grover Iterate, and then Expression~\ref{eq:amplitude-target-state} changes to:

\begin{eqnarray}
\cos{((2j+1)\theta)}
\label{eq:amplitude-target-state-cos-u}
\end{eqnarray}

For a number of Grover iterations $j \geq 0$, measuring the state described in Eq.~\ref{eq:k-applications-grover} will return the target state with probability

\begin{eqnarray}
\sin^2{((2j+1)\theta)}. 
\label{eq:probability-target-state}
\end{eqnarray}

Choosing $j = [\frac{\pi}{4\theta} - \frac{1}{2}]$ will maximally amplify the amplitude of the target state. We will denote this number of iterations by $j_\mathit{uniform}$. Sometimes $j = \lfloor\frac{\pi}{4}\sqrt{2^n}\rfloor$ is used as a simpler approximation~\cite{Grover1996}. 

\section{\label{sec:generalized-grover}Binomial Version of Grover's Algorithm}

\subsection{Binomial Grover Iterate\label{subsec:binomial-grover-iterate}}
In~\cite{Gilliam2020Optimizing, Gilliam2020Canonical} we introduced a generalization of the Amplitude Amplification algorithm~\cite{Grover1996, Brassard2000}, with the following ingredients:

\begin{itemize}
    \item Any unitary operator $A$, which creates a superposition state
    \item A \emph{property-encoding operator} $B$ that maps distinct groups of states in different ways
    \item A \emph{simple, binary-matching oracle} $O_B$, whose purpose is to match the value assigned to the selected states
    \item The \emph{diffusion operator} $D$, which multiplies the amplitude of the $\ket{0}_n$ state by $-1$
\end{itemize}
Based on the above, the Grover Iterate is now defined as $G = -ADA^{\dagger}O$, where $O=B^{\dagger}O_BB$ is the \emph{canonical oracle} that, by construction, flips the phases of the selected states, and the Amplitude Amplification routine takes the form of $BG^jA$, where $j$ is the number of times the Grover Iterate is applied.

In this paper, we introduce a new form of this type of generalization, where the operator $B$ prepares a binomial distribution before the application of the simple oracle that marks a given target state. The amplitude of a basis state will depend on the number of $1$s in its binary representation, and all states with the same number of $1$s in their binary representation will have the same amplitude.

\subsection{Binomial Grover's Search Algorithm \label{subsec:generalized-grover-search-algorithm}}
As in the standard Grover's Search algorithm, we assume that all the states in the search space are in equal superposition, created by $A=H^{\otimes n}$.
The operator $B$ converts the uniform superposition to a binomial one, where the basis states with the same number of $1$s in their binary representation will have the same amplitude. The set of distinct amplitudes is a geometric sequence.

When starting with a uniform superposition, for a given $\omega$ such that $0 \leq \omega \leq \pi$, the Grover Iterate is defined by:

$$G(\omega) = -H^{\otimes n}DH^{\otimes n}B^{\dagger}(\omega)OB(\omega)$$
where
$$B(\omega) = R_Y^{\otimes n}(\omega)H^{\otimes n} = R_Y^{\otimes n}(\pi/2 + \omega)Z^{\otimes n}$$.

When starting with a binomial superposition, created by $A = R_Y^{\otimes n}(\omega)$, we have:
$$G(\omega) = -R_Y^{\otimes n}(\omega)DR_Y^{\otimes n}(-\omega)O$$.

In both the uniform and binomial cases the following result holds:

\begin{theorem}[Binomial Amplification]
Given an $n$-qubit quantum system, a target state marked by an oracle $O$, whose binary representation has a number $0 \leq k \leq n$ of $1$s, and an angle $0 \leq \omega \leq \pi$, the amplitude of the target state after applying the modified amplitude amplification procedure with $j \geq 0$ Grover iterations $G(\omega)$  is
\begin{eqnarray}
\sin((2j + 1) \theta(\omega))
\label{eq:amplitude-target-state-generalized}
\end{eqnarray}

where 

\begin{eqnarray}
\theta(\omega) = \arcsin\left(\sin^{k}\frac{\omega}{2} \cos^{n-k}\frac{\omega}{2}\right).
\end{eqnarray}

\end{theorem}

\begin{proof}
The setup is equivalent to the one described in~\cite{Brassard2000} with $A = R_Y^{\otimes n}(\omega)$. 

For a single qubit: 
$$R_Y(\omega)\ket{0} = \cos\frac{\omega}{2}\ket{0} +  \sin\frac{\omega}{2}\ket{1}.$$

Therefore, for $n$ qubits:

$$R_Y^{\otimes n}(\omega)\ket{0}_n = \sum_{i = 0}^{2^n - 1} \sin^{k(i)}\frac{\omega}{2} \cos^{n-k(i)}\frac{\omega}{2} \ket{i}_n,$$
where $k(i)$ denotes the number of $1$s in the binary representation of $i$, also called the \emph{Hamming weight} of this representation.

Note that all amplitudes are non-negative real numbers. Following the argument in~\cite{Boyer1998}, if $\theta$ denotes the angle between a non-target state vector and the superposition state vector, then $\sin(\theta)$ equals the amplitude of the target state.

\end{proof}

Compared to the uniform version, which uses a fixed angle $\theta_{uniform}$, in the binomial version we have a parameterized angle $\theta(\omega)$ that varies in a range that includes $\theta_{uniform}$.

For a given $0 \le k \le n$, denote the amplitude of a state whose binary representation contains $k$ $1$s before the amplification process by:

$$a_k(\omega) = \sin\left(\theta(\omega)\right) = \sin^{k}\frac{\omega}{2} \cos^{n-k}\frac{\omega}{2}$$ 

It is easy to check that, for the range of $\omega$ that we are interested in:
$$a_k(\omega) \le \frac{1}{2}.$$

The amplification process is visualized in Figure~\ref{fig:unit_circle}. The angle between the non-target  state vector and the superposition state vector is initially $\theta(\omega$), and each Grover iteration increases it by $2\theta(\omega$). The best case scenario is to only perform one iteration and make this angle $\pi/2$, corresponding to a measurement probability of $1$ for the target state, and it occurs when $\theta(\omega) = \pi/6$ and the initial amplitude is $\frac{1}{2}$. 

\begin{figure}[tb]
\includegraphics[width=5cm]{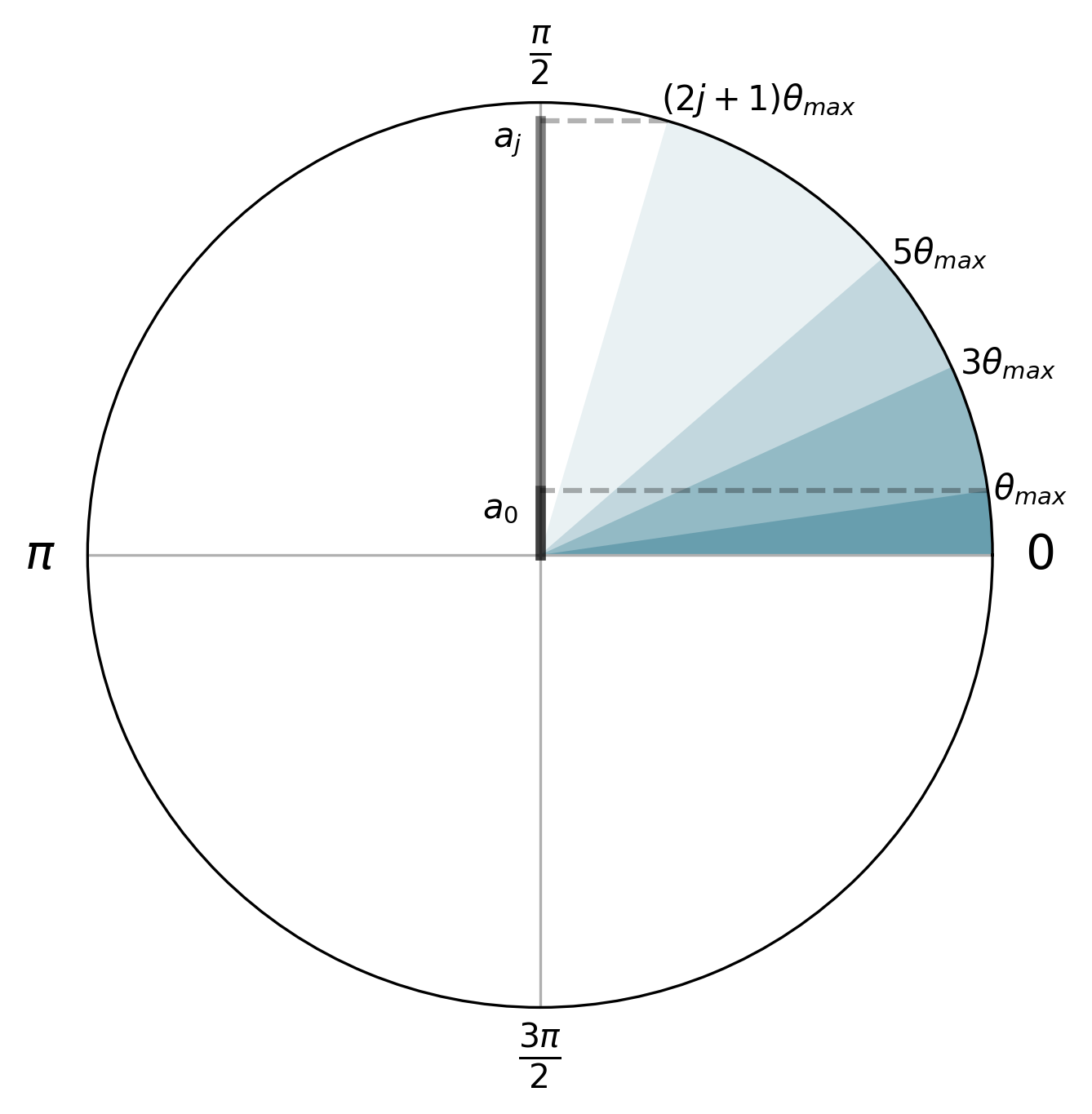}
\caption{\label{fig:unit_circle}Multiples of the parameterized angle $\theta_\mathit{max}$ plotted on a unit circle, where $j$ is the number of Grover iterations, and amplitude $a_j=\sin((2j+1)\theta_\mathit{max})$.}
\end{figure}

Intuitively, this allows for the following improvements:

\begin{itemize}
    \item Taking $\theta$ to be the upper bound of the range leads to a larger amplitude and probability of the target state after the same number of iterations as in the uniform version
    \item We can choose a parameter $\omega$ such that $\theta(\omega)$ is $\pi/2$ for a well chosen number of iterations.
\end{itemize}

Note that when $\omega=\frac{\pi}{2}$ we retrieve Expression~\ref{eq:amplitude-target-state}:

\begin{eqnarray}
\sin((2j + 1) \arcsin((\frac{1}{\sqrt{2}})^n)
\end{eqnarray}

As noted in the uniform version,  we can remove the negative sign in the definition of the Grover Iterate, and then Expression~\ref{eq:amplitude-target-state-generalized} changes to:

\begin{eqnarray}
\cos{((2j+1)\theta)}
\label{eq:amplitude-target-state-cos}
\end{eqnarray}

In order to minimize the number of Grover iterations, we are interested in maximizing the angle $\theta(\omega$). As mentioned before: $a_k(\omega) \le \frac{1}{2}.$

When $k = 0$ we can attain equality with: 
$$\omega_\mathit{max} = 2\arccos(\frac{1}{2^{1/n}}).$$

When $k = n$ we can again attain equality with: 
$$\omega_\mathit{max} = 2\arcsin(\frac{1}{2^{1/n}}).$$

This means that for $k=0$ or $k=n$ we only need one Grover iteration to amplify the probability of the target state to $1$.

For $ 0 < k < n$, if we look at the critical points of $a_k(\omega)$, whose derivative with respect to $\omega$ is:

 $$  \frac{1}{2}\sin^{k-1}\frac{\omega}{2} \cos^{n-k-1}\frac{\omega}{2} \left(k\cos^{2}\frac{\omega}{2} - (n-k)\sin^{2}\frac{\omega}{2}\right)$$

we find out that in this case:

$$\omega_\mathit{max} = 2\arctan\left(\sqrt{\frac{k}{n-k}}\right)$$.

For a balanced outcome $k = n-k$, we retrieve the angle in standard Grover: 
    $$\omega_\mathit{max} = 2\arctan(1) = \frac{\pi}{2}$$.

We will use the notation $\theta_\mathit{max} = \theta(\omega_\mathit{max})$ to refer to the upper range of $\theta(\omega)$ in the binomial version of the amplitude amplification process, and $\theta_\mathit{uniform} = \arcsin\left(\sqrt{\frac{1}{2^n}}\right)$ to refer to the angle in the uniform version.

After a number of Grover iterations, the angle between the search non-target state vector and the superposition state vector approaches or surpasses $\pi/2$. We are interested in finding this particular number of iterations, and adjust $\omega$ in order to make the angle match $\pi/2$. 

For fixed $n$ and $k$, the procedure can be broken down into 2 steps: 

\begin{itemize}
    
    \item find $j_\mathit{ideal}$, defined as the minimum $j$ such that \mbox{$(2j + 1) \theta_\mathit{max} \ge \pi/2$}:
    
    $$j_\mathit{ideal} = \lceil\frac{\pi}{4\theta_\mathit{max}} - \frac{1}{2}\rceil$$
    
    and define:
    $$\theta_\mathit{ideal} = \frac{\pi}{2(2j_\mathit{ideal}+1)}$$
    
    \item find $\omega_\mathit{ideal}$ that satisfies:
    $$\sin^{k}\frac{\omega_\mathit{ideal}}{2} \cos^{n-k}\frac{\omega_\mathit{ideal}}{2}=\sin(\theta_\mathit{ideal})$$

\end{itemize}

Applying the binomial version of the amplitude amplification process with $j_\mathit{ideal}$ iterations will lead to a measurement probability of $1$ for the search target state. Since $\theta_\mathit{max} \ge \theta_\mathit{uniform}$, the number of iterations 
$j_\mathit{ideal}$ is less than or equal to the number of iterations used in the standard Amplitude Amplification procedure to maximize the amplitude of the search target state.

\section{\label{sec:analysis}Additional Analysis and Examples}

In this section we analyze the impact of choosing different partitions and numbers of iterations on the final result. We will denote by $n$ the number of qubits in a quantum system, and by $k$ the number of $1$s in the binary representation of a search target state.

In the extreme cases when $k = 0$ or $k = n$, the binomial version of the amplification process only needs one iteration.

With the binomial version of the amplitude amplification process, we have seen that for a given $0 \le k \le n$, we can amplify the probability of a state whose binary representation has $k$ $1$s to the maximum possible of $1$. 
The amplification takes $j_\mathit{ideal}$ iterations. 
In Table 1, we can see how $j_\mathit{ideal}$ compares to $j_\mathit{uniform}$, the best number of iterations in the standard amplitude amplification process, for an 8-qubit circuit.

\begin{center}
\begin{tabular}{||c | c | c||}
\hline
$k$ &  $j_\mathit{uniform}$ & $j_\mathit{ideal}$ \\ [0.5ex]
\hline\hline
0 & 12 & 1 \\
\hline
1 & 12 & 3 \\
\hline
2 & 12 & 7 \\
\hline
3 & 12 & 11 \\
\hline
4 & 12 & 12 \\
\hline
5 & 12 & 11 \\
\hline
6 & 12 & 7 \\
\hline
7 & 12 & 3 \\
\hline
8 & 12 & 1 \\
\hline
\end{tabular}
\end{center}
TABLE 1. A comparison of the number of iterations required to measure the target state with maximum probability for an 8-qubit state.
\newline
\newline
Fig.~\ref{fig:qubits_v_depth} shows how the number of iterations required to get the maximum probability scales in regards to $n$.

\begin{figure}[t]
\includegraphics[width=8cm]{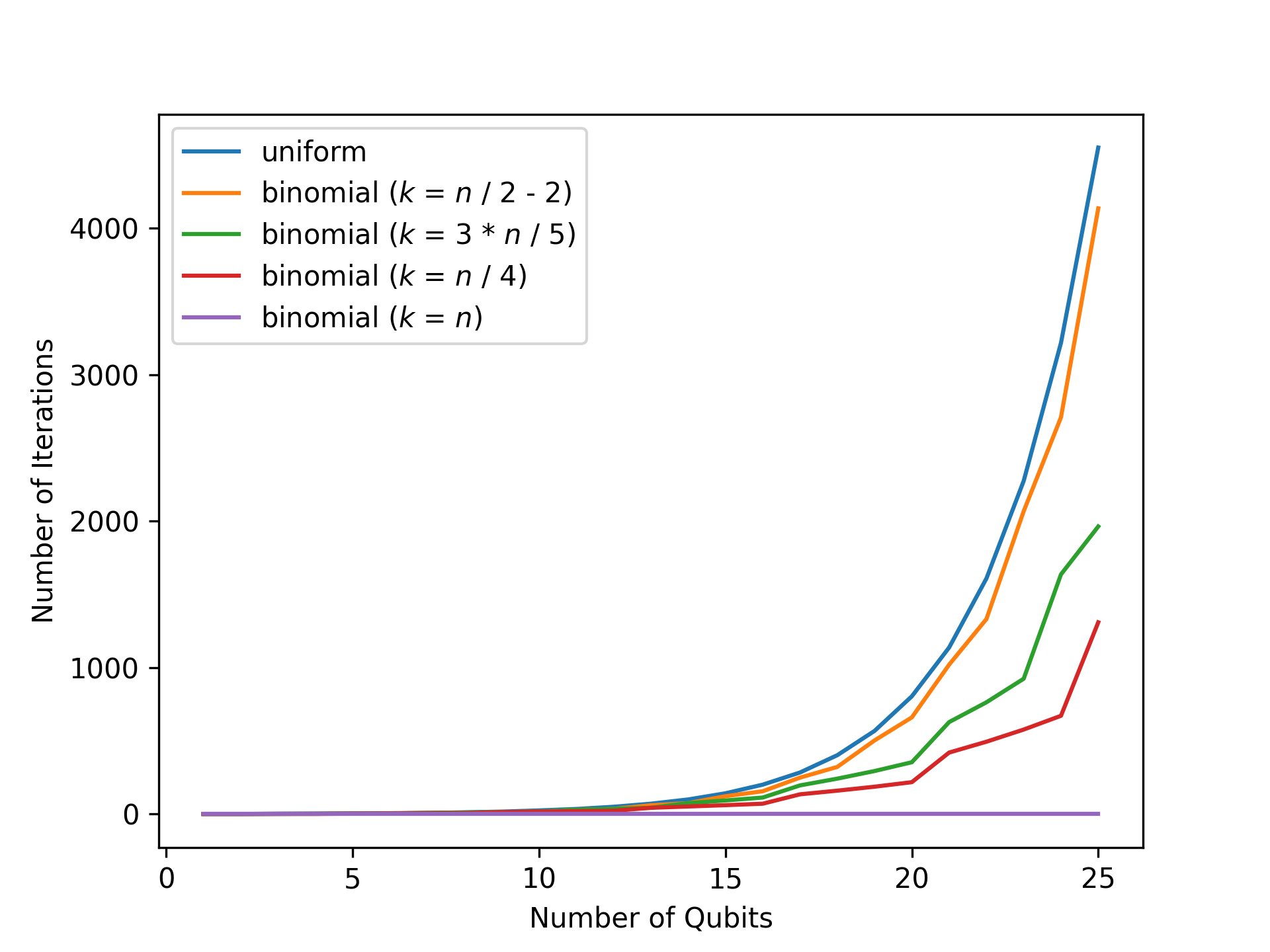}
\caption{\label{fig:qubits_v_depth}A comparison of the number of Grover iterations required to reach the maximum probability with respect to the number of qubits. Results for uniform Grover and binomial Grover with various $k$ values are shown.}
\end{figure}

Given that on current quantum hardware the depth of a circuit with more than $2$ or $3$ Grover iterations is prohibitive, binomial amplification can provide a clear advantage for certain $k$s. 
In this case, we have no choice but to fix the maximum number of iterations.
Continuing the example in Table 1, we can see how the probabilities compare in the uniform and binomial versions for a given number of iterations in Fig.~\ref{fig:iteration_vs_probability}.

\begin{figure}[H]
\includegraphics[width=8.5cm]{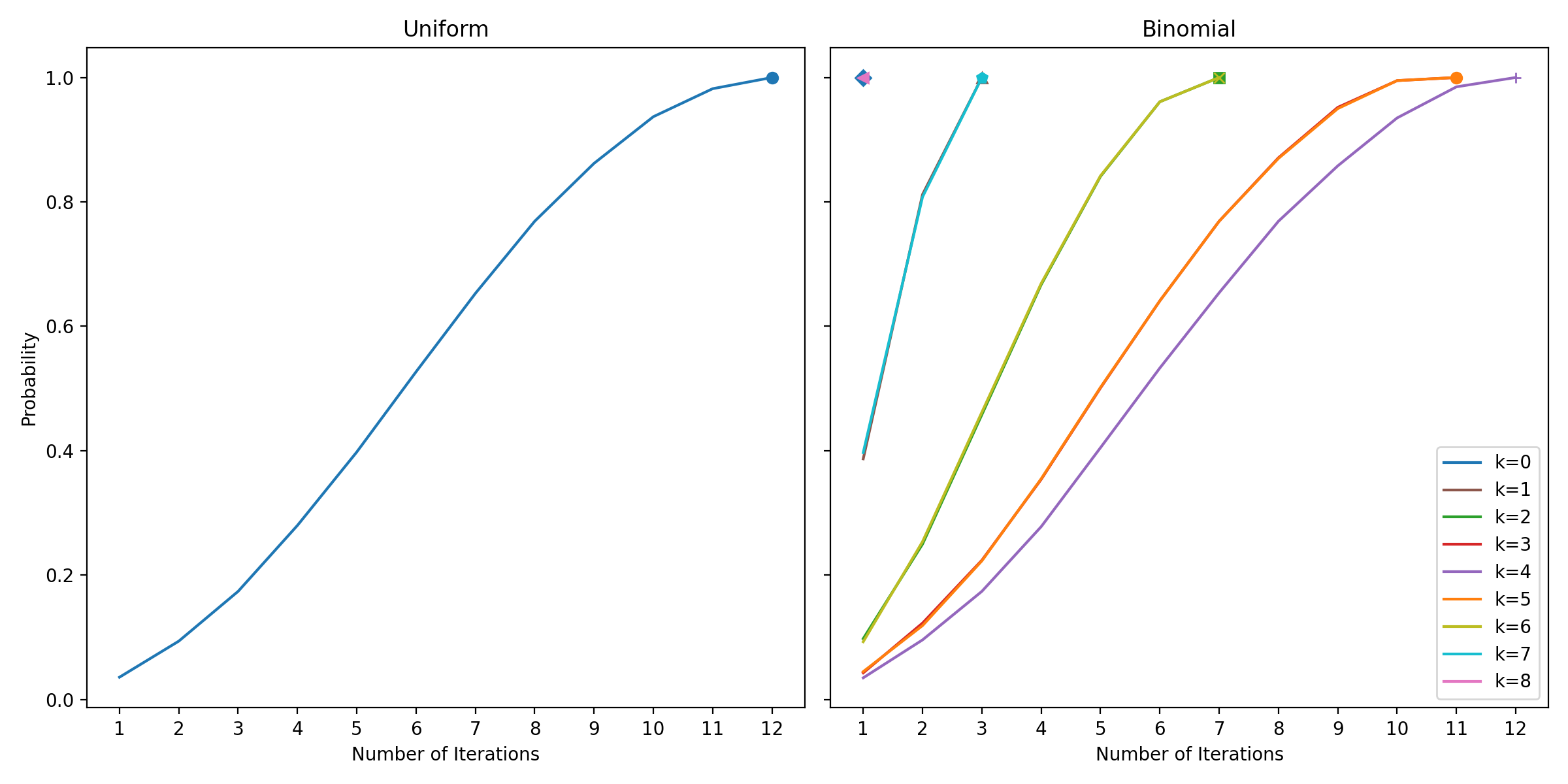}
\caption{\label{fig:iteration_vs_probability}A comparison of uniform and binomial Grover for $8$ qubits, where $k$ is the number of $1$s in the binary label of the outcome.}
\end{figure}

As a specific example, consider the case with $n=8$ qubits and a target state of $10000001$ (the binary representation of $129$, containing $k=2$ $1$s), with a limitation of $j=6$ Grover iterations. We can calculate $\theta_{max}$ from $\omega_{max}$. As seen in Fig.~\ref{fig:grover_compare_n_8}, we get a higher probability when compared to the uniform approach. In Sec.~\ref{sec:experimental-results}, we show additional examples that handle this constraint.

\begin{figure}[htb]
\includegraphics[width=6cm]{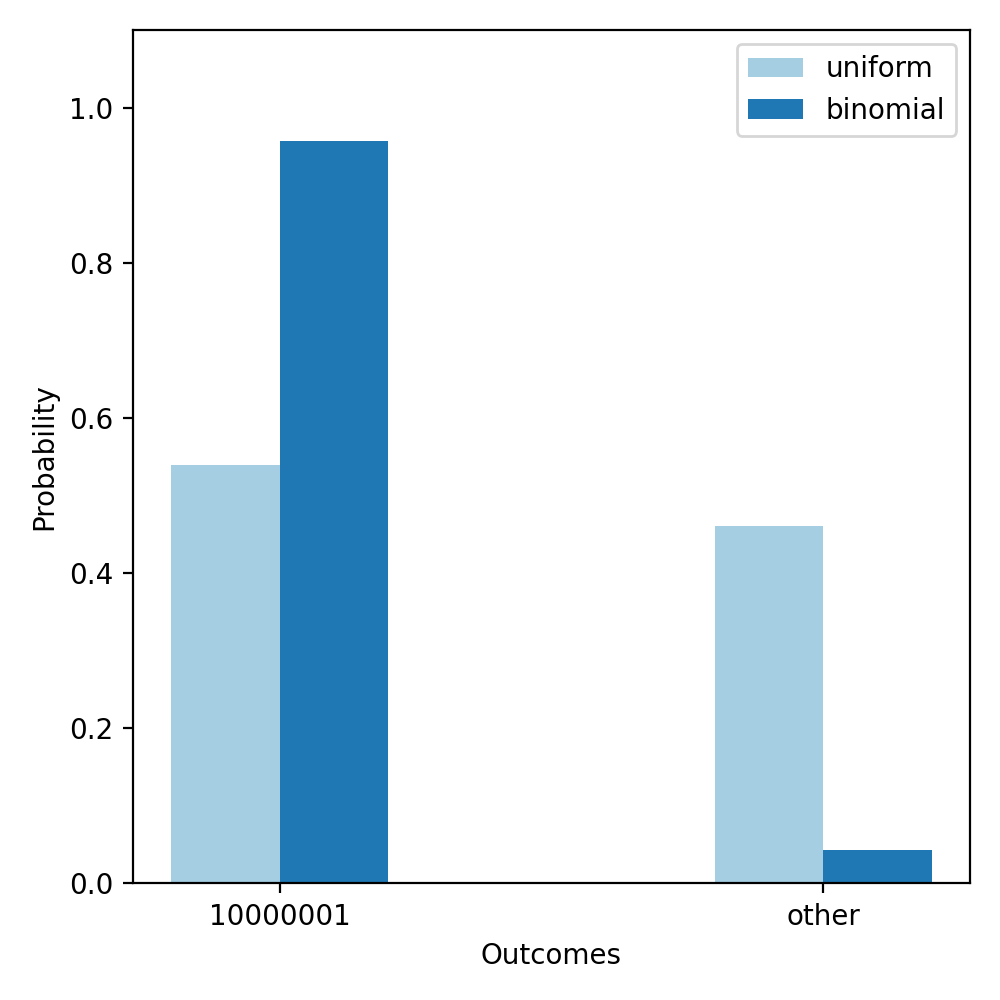}
\caption{\label{fig:grover_compare_n_8}A comparison of uniform and binomial Grover on a state with $8$ qubits, $6$ iterations, and search target $10000001$. The probability of non-target states are summed together for easier comparison.}
\end{figure}

Another strategy for dealing with the limitation on the number of Grover iterations is to use an adaptive version~\cite{Durr1996, Bulger2003, Baritompa2005, Gilliam2019}, where a random number of iterations is applied.
The binomial version allows one to also randomly choose the partition to which search target states are assumed to belong.

For the case of multiple search state targets, instead of a single partition, we have to allow for multiple ones.
Let us denote by $T$ the set of target states, and $k(i)$ as the Hamming weight (number of $1$s in the binary representation) of a state $\ket{i}$.
The angle $\theta$ between the superposition state vector and the non-target state is given by:

$$
\theta(\omega) = \arcsin\left(\sum_{i \in T}\sin^{k(i)}\frac{\omega}{2} \cos^{n-k(i)}\frac{\omega}{2}\right).
$$

Note that for $\omega=\frac{\pi}{4}$, we retrieve the multi-target uniform version where $\theta = \arcsin(|T|\frac{1}{\sqrt{2^n}})$.

\section{\label{sec:experimental-results}Experimental Results}

In this section, we will apply the concepts described in the previous sections to examples of quantum search. Each experiment described in this section was run on real quantum hardware---the Honeywell System Model HØ trapped-ion quantum computer with quantum volume 64~\cite{HoneywellArchitecture, QuantumVolume}---with one execution consisting of 1024 shots. No calibration or noise mitigation was applied.

\subsection{\label{subsec:set-search}Set Element Search}
Let us consider an example of set search using binomial Grover, and compare the results with the uniform approach.
In Figure~\ref{fig:set-search} we show a $4$-qubit circuit implementing the binomial version of Grover's Search algorithm with $1$ iteration, where the target state is $1101$ (the binary representation of $13$).

\begin{figure}[htb]
\includegraphics[width=8cm]{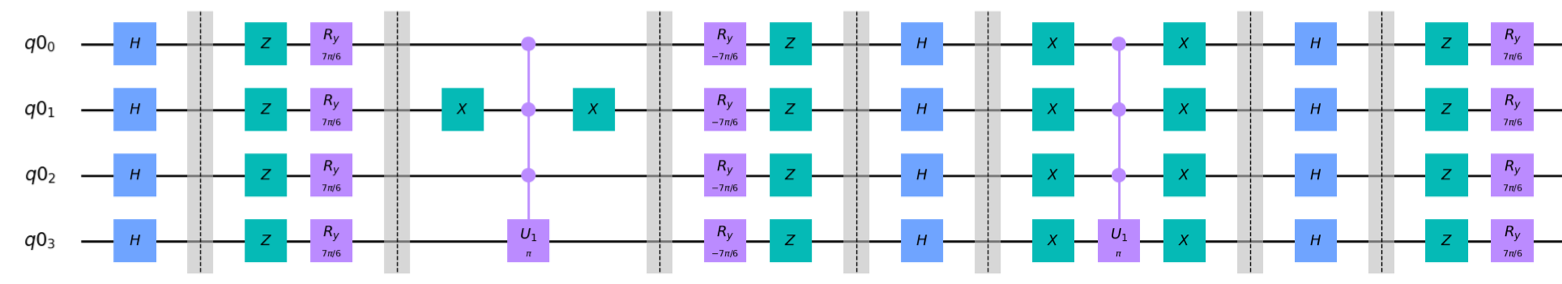}
\caption{\label{fig:set-search}A circuit implementing the binomial version of set element search with $1$ iteration and search target $1101$.}
\end{figure}

Note that we assume an equal superposition to start, and then convert to a binomial distribution.
The rotation angle of $R_Y$ is dependent on the number of $1$s in the label of the target state, as described in Sec.~\ref{subsec:generalized-grover-search-algorithm}.
For $10$ out of $16$ values we can improve the probability using a binomial approach, at the expense of knowing more information about the search target state. The result of this circuit can be seen in Fig.~\ref{fig:set-search-n-4}.

\begin{figure}[tb]
\includegraphics[width=8cm]{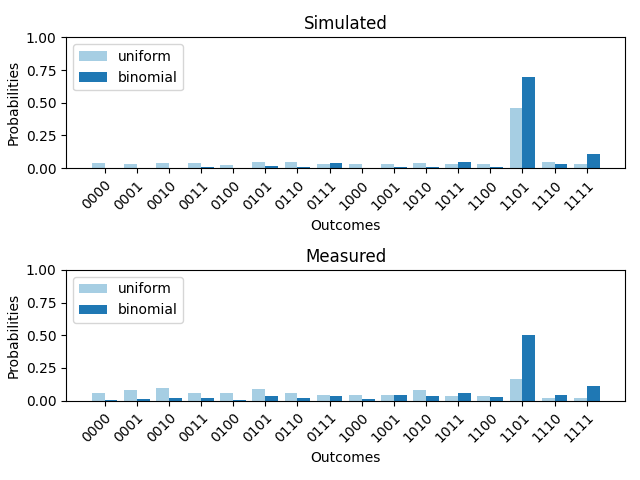}
\caption{\label{fig:set-search-n-4}A comparison of the uniform and binomial versions of set search, on an example with $4$ qubits, with search target $1101$ and $1$ iteration, using a simulator (top) and real quantum hardware (bottom).}
\end{figure}

\subsection{\label{subsec:array-retrieval}Array Element Retrieval}

In~\cite{Gilliam2020Optimizing}, we discuss array search, where we can use a Quantum Dictionary~\cite{Gilliam2019, Gonciulea2019} to index values by a separate register that holds all indices in superposition.
In order to retrieve a value, we need to perform a search on the index register.
The classical version of associate arrays uses direct memory access, or hashing, to speedup value retrieval.
In quantum computing, amplitude amplification is playing a similar role, using an oracle that will mark a desired index.
The type of superposition in the index register will affect the speed of the retrieval.
Here we are comparing the performance of a uniform and binomial superposition types in the index register.

As an example, consider a Quantum Dictionary circuit that encodes the partial and total sums of the elements in the array $[1, -1, 1]$.
The encoding of the array using the Quantum Dictionary pattern needs a $3$-qubit index register and a $2$-qubit value register, that can represent negative values using two's complement.
Recall that the Quantum Dictionary pattern encodes the array values as periodic signals, whose interference leads to the sums of all possible subsets.
A subset is represented by a binary string in the index register, and the corresponding sum is encoded in the value register.

In Figure~\ref{fig:qd-key-circuit}, we consider the case when the target state is $110$, the binary representation of $6$. The result of this circuit be seen in Figure~\ref{fig:array-retrieval-m-6}.

\begin{figure}[tb]
\includegraphics[width=8cm]{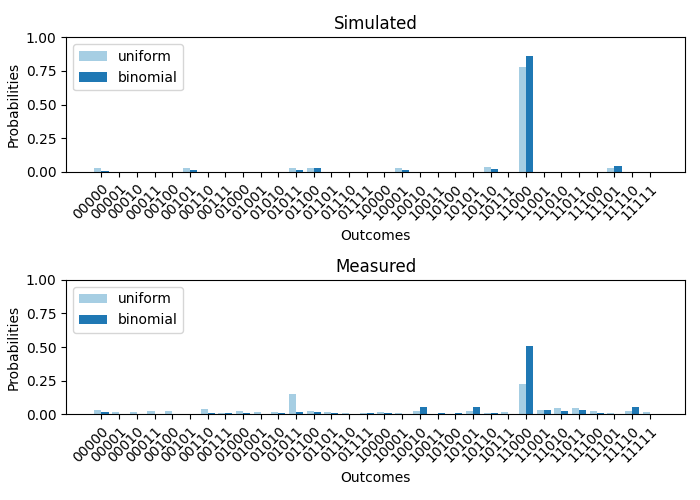}
\caption{\label{fig:array-retrieval-m-6}A comparison of the uniform and binomial versions of array element retrieval, on an example with $3$ index qubits, $2$ value qubits, and search target $110$ using a simulator (top) and real quantum hardware (bottom).}
\end{figure} 

\subsection{\label{subsec:array-search}Array Element Search}

When searching for an array element, we need an oracle that will mark a desired value.
We also need to correctly estimate the number of times the Grover Iterate has to be applied, in order to maximize the probability of measuring the desired value.
Alternatively, we can use the Grover Adaptive Search~\cite{Durr1996, Bulger2003, Baritompa2005} algorithm in order to avoid this estimation.
The adaptive approach with a small number of iterations may be the only option in the current era of quantum hardware, where the circuit depth is limited.

Using a binomial superposition in the index register allows for the random amplitude amplification of states based on the Hamming weight of the indices.
For example, one can first choose a binomial superposition that will favor the indices with more $1$s than $0$s in their binary representation, and another that will favor those with less $1$s than $0$s.

This approach is similar to classical binary search algorithms.
As an example, if we have $11$ qubits, we can choose the binomial distribution given by $\omega=\frac{15}{32}\pi$.
This will actually increase the initial amplitudes of all targets with more $0$s than $1$s in their binary representation compared to the uniform version, as seen in Table 2.

\begin{center}
\begin{tabular}{||c | c | c||}
\hline
k &  binomial & uniform \\ [0.5ex]
\hline\hline
0 & 0.03696  & 0.0221 \\
\hline
1 & 0.03349  & 0.0221 \\
\hline
2 & 0.03036  & 0.0221 \\
\hline
3 & 0.02751  & 0.0221 \\
\hline
4 & 0.02494  & 0.0221 \\
\hline
5 & 0.0226  & 0.0221 \\
\hline
\end{tabular}
\end{center}
TABLE 2. A table comparing the initial amplitudes for binomial version defined by $\omega=\frac{15}{32}\pi$, compared to the uniform version for various $k$s (where $k$ is the Hamming weight of an index).
\newline

Similarly, if we choose the binomial distribution given by $\omega=\frac{17}{32}\pi$, we maximize the the initial amplitudes of all targets with more $1$s than $0$s in their binary representation, as seen in Table 3.

\begin{center}
\begin{tabular}{||c | c | c||}
\hline
k &  binomial & uniform \\ [0.5ex]
\hline\hline
6 & 0.0226  & 0.0221 \\
\hline
7 & 0.02494  & 0.0221 \\
\hline
8 & 0.02751  & 0.0221 \\
\hline
9 & 0.03036  & 0.0221 \\
\hline
10 & 0.03349  & 0.0221 \\
\hline
11 & 0.03696  & 0.0221 \\
\hline
\end{tabular}
\end{center}
TABLE 3. A table comparing the initial amplitudes for binomial version defined by $\omega=\frac{17}{32}\pi$, compared to the uniform version for various $k$s (where $k$ is the Hamming weight of an index).
\newline
\newline

A higher or lower amplification can be obtained by adjusting the number of $\omega$ parameters, and their values.

We ran both the uniform and binomial versions of array element search on a real quantum computer, as seen in Fig.~\ref{fig:array_search_real}. 
Here, we are searching for negative values using the same array as in Sec.~\ref{subsec:array-retrieval}.

\begin{figure}[htb]
\includegraphics[width=8cm]{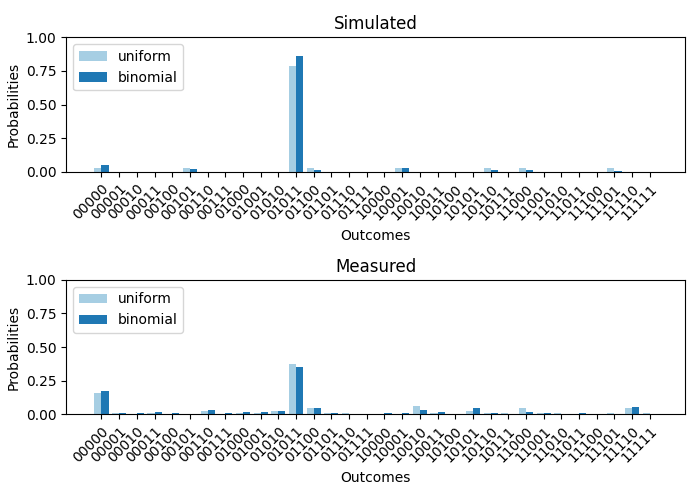}
\caption{\label{fig:array_search_real}A comparison of the uniform and binomial versions of array element search, on an example with $3$ index qubits and $2$ value qubits, searching for negative values using a simulator (top) and real quantum hardware (bottom).}
\end{figure} 

\section{\label{sec:conclusion} Conclusion}

We have outlined the method of splitting the basis states into partitions and amplifying the amplitudes of the states in each partition in a non-uniform way, before repeatedly applying a Grover Iterate specific to a given search target state. When the partition the target state belongs to is sufficiently amplified, the number of times the Grover Iterate needs to be applied can be reduced. 

This method can provide an additional parameter that can be used in optimizing quantum search. In particular, in the adaptive version of Grover's Search algorithm, in addition to varying the number of times the Grover Iterate is applied, one can also vary a level of non-uniform amplitude pre-amplification.

We have provided the numerical analysis for the partitioning of basis states based on the number of $1$s in their binary distribution, which essentially leads to a binomial distribution for the partition amplitudes.

We have shown how the method can be applied to set search, array element retrieval, and array element search. We have also successfully validated the results on a real quantum computer, and we believe these to be one of the earliest successful quantum search experiments on 4 and 5 qubits using a quantum computer.

\section*{Acknowledgements}
We would like to thank the Honeywell Quantum Solutions team for their help on the execution of our experiments on the Honeywell trapped-ion quantum computer.

\section*{Disclaimer}
This paper was prepared for information purposes by the Future Lab for Applied Research and Engineering (FLARE) Group of JPMorgan Chase \& Co. and its affiliates, and is not a product of the Research Department of JPMorgan Chase \& Co. 
JPMorgan Chase \& Co. makes no explicit or implied representation and warranty, and accepts no liability, for the completeness, accuracy or reliability of information, or the legal, compliance, tax or accounting effects of matters contained herein.
This document is not intended as investment research or investment advice, or a recommendation, offer or solicitation for the purchase or sale of any security, financial instrument, financial product or service, or to be used in any way for evaluating the merits of participating in any transaction.

\begin{figure*}[htb]
\includegraphics[width=16cm]{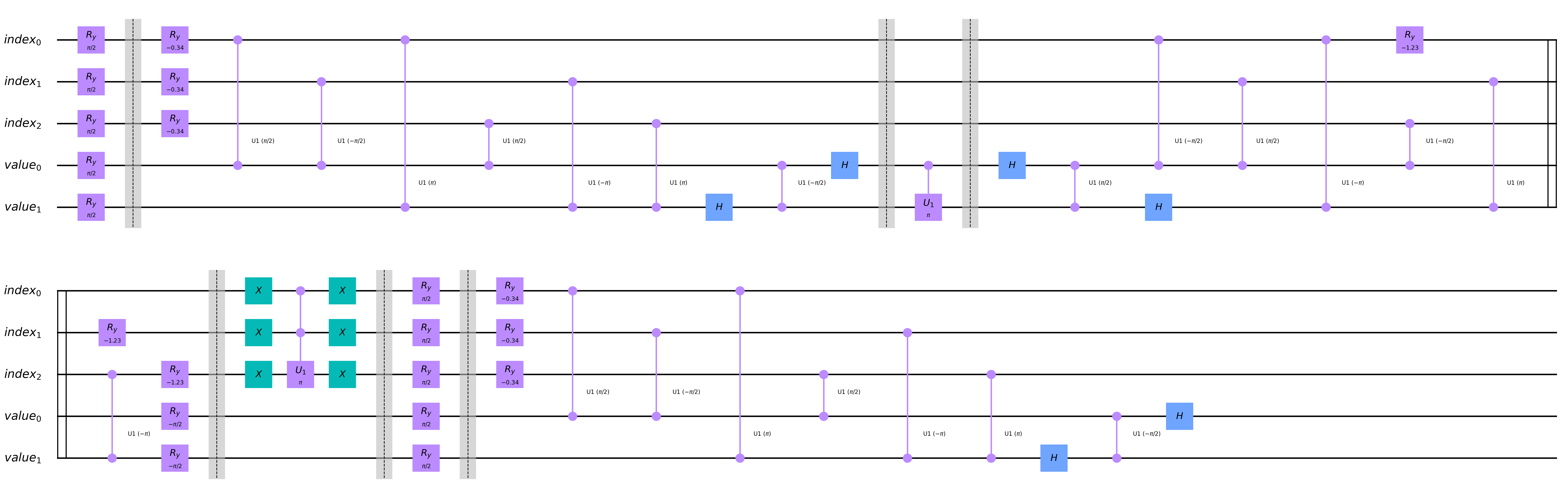}
\caption{\label{fig:qd-key-circuit}A circuit implementing the binomial version of array element retrieval with $1$ iteration and search target $110$.}
\end{figure*}

\bibliography{main}
 
\end{document}